\documentclass[12pt, twoside]{article}
\usepackage{mathrsfs}
\usepackage{amsfonts}
\usepackage{setspace}
\usepackage{amsmath,amssymb}
\usepackage{multicol}
\usepackage{color}
\usepackage{cite}
\newtheorem{theorem}{Theorem}[section]
\newtheorem{lemma}[theorem]{Lemma}

\newtheorem{definition}[theorem]{Definition}
\newtheorem{remark}[theorem]{Remark}

\numberwithin{equation}{section}
\newenvironment{proof}[1][Proof]{\noindent\textbf{#1. }}{\hfill $\Box$}
\allowdisplaybreaks

\makeatletter
\begin{document}
\author{Guang-Hui Zheng\thanks{Email:
zhgh1980@163.com(Guang-Hui Zheng).}\\
{\small College of Mathematics and Econometrics, Hunan University, } \\
{\small Changsha 410082, Hunan Province, P.R. China}\\
}
\title{\textbf{\Large Mathematical analysis of plasmonic resonance for 2-D photonic crystal}}
\maketitle
\begin{abstract}
In this article, we study the plasmonic resonance of infinite photonic crystal mounted by the double negative nanoparticles in two dimensions. The corresponding physical model is described by the Helmholz equation with so called Bloch wave condition in a periodic domain. By using the quasi-periodic layer potential techniques and the spectral theorem of quasi-periodic Neumann-Poincar{\'e} operator, the quasi-static expansion of the near field in the presence of nanoparticles is derived. Furthermore, when the magnetic permeability of nanoparticles satisfies the Drude model, we give the conditions under which the plasmonic resonance occurs, and the rate of blow up of near field energy with respect to nanoparticle's bulk electron relaxation rate and filling factor are also obtained. It indicates that one can appropriately control the bulk electron relaxation rate or filling factor of nanoparticle in photonic crystal structure such that the near field energy attains its maximum, and enhancing the efficiency of energy utilization.


\end{abstract}

\section{Introduction}

\noindent
Plasmonic resonance has been applied recently in various scientific fields, such as enhancing the brightness of light,
confining strong electromagnetic fields, medical therapy, invisibility cloaking, biomedical imaging \cite{ACKLM5,CKKL7,ACKLM4,SC1,KLSW8,JLEE3,BGQ2,LLL9,WN10,BS11,ADM6,N19,N29,NN30} and so on.

In \cite{ACKLM5}, Ammari et al. give a necessary and sufficient condition for electromagnetic power dissipation to blow up
as the loss parameter of the plasmonic material goes to zero. Moreover, under some additional conditions,
the cloaking due to anomalous localized resonance (CALR) will happen. The confocal ellipses case for plasmonic resonance and CALR
are studied by Chung et al. \cite{ACKLM4}. Ando and Kang \cite{AK12} investigate the plasmonic resonance for conductivity equation on smooth domain. The plasmonic resonance analysis for Helmholtz equation with finite frequencies is discussed by Ando, Kang and Liu
\cite{AKL13}. Ammari et al. \cite{AMRZ14} give a mathematical analysis of plasmonic resonance for Helmholtz equation and apply it in the scattering and absorption enhancements, super-resolution and super-focusing. Afterward, they extend the relative results to full Maxwell equations (see \cite{ARYZ15}). Recently, in \cite{ARWYZ16} and \cite{AFGLZ17}, Ammari et al. consider the scattering problem for periodic plasmonic nanoparticles mounted on a perfectly conducting sheet.

In this paper, applying the method given by \cite{AMRZ14} and \cite{AKL13}, we give the analysis of plasmonic resonance of two dimensions infinite photonic crystal mounted by the double negative nanoparticles. Different from the model in \cite{AMRZ14}, the corresponding physical model here is described by the Helmholz equation with the Bloch wave condition (see \cite{FKH18}) in a periodic domain, and we focus on the wave propagation behavior inside the photonic crystal structure. It is well known that the plasmonic surface wave (Bloch wave) propagate along the interface between nanoparticles and background media, and the field intensity falls off evanescently perpendicular to the interface \cite{JMW24}. Moreover, as the nanoparticles in photonic crystal has a negative electromagnetic parameter (such as double negative material or left-hand material) and small size in comparison with incidence wavelength (for example, in low frequency region), the resonance phenomena is often observed both experimentally and numerically \cite{MFZ31,MZ32}. By using the quasi-periodic layer potential techniques and the spectral theorem of quasi-periodic Neumann-Poincar{\'e} operator, we derive the quasi-static expansion formula for the near field of nanoparticles. Furthermore, in the quasi-static regime (i.e. frequency $\omega$ is small enough), the conditions closely related to the eigenvalue of quasi-periodic Neumann-Poincar{\'e} operator are obtained, under which the quasi-periodic plasmonic resonance happens. Then based on the Drude model that the magnetic permeability of nanoparticles satisfy, we get the rate of blow up of near field energy with respect to nanoparticle's bulk electron relaxation rate and filling factor. It shows that
one can control the bulk electron relaxation rate or filling factor of nanoparticle in photonic crystal structure such that the near field energy attains its maximum, and improving the efficiency of energy utilization (such as enhancing the brightness of light, confining strong electromagnetic fields).

The paper is organized as follows. In section 2, we introduce the problem formulation for plasmonic resonance in two dimensions infinite photonic crystal by using the quasi-periodic layer potential. In section 3 we derive the asymptotic expansion formula of the near field in
quasi-static regime. Based on Drude model, the rate of blow up of near field energy with respect to nanoparticle's bulk electron relaxation rate and filling factor are given in section 4. Finally, we give a conclusion in section 5.

\section{Problem formulation}

\noindent

The photonic crystal we consider in this paper consists of a homogeneous background medium which is perforated
by an array of arbitrary-shaped plasmonic nanoparticles periodically along each of the two orthogonal
coordinate axes in $\Bbb{R}^2$. We assume that the structure has unit periodicity and define the unit cell $Y=[0, 1]^2$.
In the unit cell $Y$, the nanoparticle occupying a bounded and simply connected domain $D\subset Y$ whose boundary $\partial D$
is $\mathcal{C}^{1,\alpha}$ for some $\alpha\in(0,1)$ which is characterized by electric permittivity $\varepsilon_c$ and magnetic permeability $\mu_c$, while the homogeneous medium $Y\setminus\overline{D}$ is characterized by electric permittivity $\varepsilon_m$ and magnetic permeability $\mu_m$. In general, the propagation of light in the photonic crystal is described by the
Maxwell equations. Here, we only focus on the transverse magnetic (TM) case, whose governing equations are reduced to Helmholtz equations.
Furthermore, assume that the nanoparticles is dispersive, i.e. $\varepsilon_c$ and $\mu_c$ are may depend on the frequency $\omega$. Let $\Re\varepsilon_c<0$, $\Im\varepsilon_c>0$, $\Re\mu_c<0$, $\Im\mu_c>0$ and define
\begin{equation}\label{2.1}
k_c=\omega\sqrt{\varepsilon_c\mu_c},\ \ k_m=\omega\sqrt{\varepsilon_m\mu_m},
\end{equation}
and
\begin{equation}\label{2.2}
\varepsilon_D=\varepsilon_c\chi(D)+\varepsilon_m\chi(Y\setminus\overline{D}),\ \ \mu_D=\mu_c\chi(D)+\mu_m\chi(Y\setminus\overline{D}),
\end{equation}
where $\chi$ denotes the characteristic function. Notice that, in nano-metal materials, the electric permittivity and magnetic permeability satisfy $\Re\varepsilon_c<0$, $\Re\mu_c<0$ is called double negative material, which shows several unusual properties, such as, counter directance between group velocity and phase vector, negative index of refraction, reverse Doppler and Cherenkov effects \cite{ZSK20,V21,SPW22}. The positive imaginary parts of electric permittivity and magnetic permeability denotes the dissipation of plasmonic nanoparticles. Moreover, we assume that $\varepsilon_m$, $\mu_m$ are real and strictly positive.

The propagation of wave in the photonic crystal with dipole source can be modeled by the following Helmholtz equations in infinite periodic structure \cite{AKL25}
\begin{align}
\begin{cases}\label{2.3}
\nabla\cdot\frac{1}{\mu_D}\nabla u+\omega^2\varepsilon_D u=a\cdot\nabla\delta_z,\ \ \ \text{in}\ Y\setminus\partial D,\\
u|_+=u|_-,\ \ \ \text{on}\ \partial D,\\
\frac{1}{\mu_m}\frac{\partial
u}{\partial\nu}\big|_+=\frac{1}{\mu_c}\frac{\partial
u}{\partial\nu}\big|_-,\ \ \ \text{on}\ \partial D,\\
e^{-i\alpha\cdot x}u\ \ \text{is 1-periodic in}\ \mathbb{R}^2\ \ \ \text{(Bloch wave condition)},
\end{cases}
\end{align}
where $a\in\Bbb{R}^2$ is a constant vector and $\delta_z$ is the Dirac function at $z\in Y\setminus\overline{D}$, and $a\cdot\nabla\delta_z$ is a dipole source (see \cite{AKL13,LZ23}). The quasi-momentum $\alpha\in B=(0,2\pi)^2$, $B$ is called Brillouin zone \cite{JMW24}. In particular, when $\alpha=0$, it is correspond to the periodic case. The last periodic condition in (\ref{2.3}) is called Bloch wave condition, which characterizes the propagation behaviour of wave in the photonic crystal.

Let $G^{\alpha,k}(x,y)$ be the two-dimensional quasi-periodic Green's function for the Helmholtz equation, which is given by \cite{AKL25}
\begin{equation}\label{2.4}
G^{\alpha,k}(x,y)=-\frac{i}{4}\sum_{n\in\mathbb{Z}^2}H_0^{(1)}(k|x-n-y|)e^{in\cdot\alpha},
\end{equation}
where $H_0^{(1)}$ is the Hankel function of the first kind of order $0$.

The quasi-periodic single layer potential of $\phi\in H^{-\frac{1}{2}}(\partial D)$ for the Helmholtz equation is defined by \cite{AKL25}
\begin{equation}\label{2.5}
\mathcal {S}^{\alpha,k}[\phi](x)=\int_{\partial D}G^{\alpha,k}(x,y)\phi(y)d\sigma(y),\ \ \ x\in \mathbb{R}^2,
\end{equation}
and the following jump formula holds
\begin{equation}\label{2.6}
\frac{\partial(\mathcal{S}^{\alpha,k}[\phi])}{\partial\nu}\bigg|_{\pm}(x)=\left(\pm\frac{1}{2}\mathcal
{I}+(\mathcal {K}^{-\alpha,k})^*\right)[\phi](x),\ \ \ a.e.\ x\in\partial D,
\end{equation}
where
\begin{equation}\label{2.7}
(\mathcal{K}^{-\alpha,k})^*[\phi](x)=\int_{\partial D}\frac{\partial G^{\alpha,k}(x,y)}{\partial\nu(x)}\phi(y)d\sigma(y).
\end{equation}

By using the quasi-periodic single layer potential (\ref{2.6}) and jump formula (\ref{2.7}), we can get the solution representation formula of (\ref{2.3})
\begin{theorem}\label{thm1}
Assume that $k_c^2$ is not an eigenvalue of $-\Delta$ in $Y\setminus
\overline{D}$ with Dirichlet boundary condition on $\partial D$ and
the $\alpha$-quasi-periodic condition on $\partial Y$ and $k_m^2$ is not an
eigenvalue of $-\Delta$ in $D$ with Dirichlet boundary condition.
Then the solution $u$ to (\ref{2.3}) can be represented as
\begin{align}\label{2.8}
u(x)=
\begin{cases}
\mathcal {S}^{\alpha,k_c}[\phi](x),\ \ \ x\in D,\\
F_z(x)+\mathcal {S}^{\alpha,k_m}[\psi](x),\ \ \ x\in
Y\setminus\overline{D},
\end{cases}
\end{align}
where
\begin{equation}\label{2.9}
F_z(x)=\int_{\Bbb{R}^2}G^{\alpha,k_m}(x,y)a\cdot\nabla_y\delta_z(y)dy=a\cdot\nabla_x G^{\alpha,k_m}(x,z),
\end{equation}
and $(\phi,\psi)\in H^{-\frac{1}{2}}(\partial D)\times H^{-\frac{1}{2}}(\partial D)$ satisfy the following integral system
\begin{align}\label{2.10}
\begin{cases}
\mathcal {S}^{\alpha,k_c}[\phi]-\mathcal {S}^{\alpha,k_m}[\psi]=0,\ \ \ \text{on}\ \partial D,\\
\frac{1}{\mu_c}\left(-\frac{1}{2}\mathcal {I}+(\mathcal
{K}^{-\alpha,k_c})^*\right)[\phi]
-\frac{1}{\mu_m}\left(\frac{1}{2}\mathcal {I}+(\mathcal
{K}^{-\alpha,k_m})^*\right)[\psi] =\frac{1}{\mu_m}\frac{\partial F_z}{\partial\nu},\ \ \ \text{on}\ \partial D.
\end{cases}
\end{align}
Moreover, the mapping
$u\rightarrow(\phi,\psi)$ from solutions of (\ref{2.3}) to solutions of (\ref{2.10}) is one-to-one.
\end{theorem}

The proof is very similar to the one in \cite{AKL25} (Section 8.6 Theorem 8.12).

We focus on the near field behavior in the quasi-static regime, i.e., the frequency $\omega\ll1$. From Lemma \ref{lem3.3}, Lemma \ref{lem1} and Remark \ref{rem3.5} in Section 3, it finds that, for $\omega$ small enough, $\mathcal {S}^{\alpha,k}$ is invertible. Therefore, the conditions of Theorem \ref{thm1} are satisfied and the first equation in (\ref{2.10}) becomes as
\begin{equation}\label{2.11}
\phi=\left(\mathcal {S}^{\alpha,k_c}\right)^{-1}\mathcal {S}^{\alpha,k_m}[\psi].
\end{equation}
Then, from the second equation in (\ref{2.10}), we have that
\begin{equation}\label{2.12}
\mathcal {A}(\omega)[\psi]=f,
\end{equation}
where
\begin{align}\label{2.13}
\mathcal {A}(\omega)=&\frac{1}{\mu_m}\left(\frac{1}{2}\mathcal
{I}+(\mathcal {K}^{-\alpha,k_m})^*\right)
+\frac{1}{\mu_c}\left(\frac{1}{2}\mathcal {I}-(\mathcal
{K}^{-\alpha,k_c})^*\right)(\mathcal {S}^{\alpha,k_c})^{-1}\mathcal
{S}^{\alpha,k_m},\\
f=&-\frac{1}{\mu_m}\frac{\partial F_z}{\partial\nu}.\label{2.14}
\end{align}
Clearly,
\begin{align}
\mathcal {A}(0)=\mathcal
{A}_0=&\frac{1}{\mu_m}\left(\frac{1}{2}\mathcal {I}+(\mathcal
{K}^{-\alpha,0})^*\right) +\frac{1}{\mu_c}\left(\frac{1}{2}\mathcal
{I}-(\mathcal
{K}^{-\alpha,0})^*\right)\nonumber\\
=&\frac{1}{2}\left(\frac{1}{\mu_m}+\frac{1}{\mu_c}\right)\mathcal
{I}+\left(\frac{1}{\mu_m}-\frac{1}{\mu_c}\right)(\mathcal
{K}^{-\alpha,0})^*,\label{2.15}
\end{align}
where $(\mathcal {K}^{-\alpha,0})^*$ is called $\alpha$-quasi-periodic Neumann-Poincar{\'e} operator.

\section{Asymptotic expansion of the near field}

In order to solve the operator equation (\ref{2.12}), we first give some basic facts about the quasi-periodic Neumann-Poincar{\'e} operator
$(\mathcal {K}^{-\alpha,0})^*$.

\begin{lemma}\label{lem3.1}
(Calder{\'o}n identity): $\mathcal {K}^{-\alpha,0}\mathcal
{S}^{\alpha,0}=\mathcal {S}^{\alpha,0}(\mathcal {K}^{-\alpha,0})^*$.
\end{lemma}
The proof is very similar to the classical one in \cite{AGJKLSW28} and so is omitted.

Recall that $\mathcal {S}^{\alpha,0}:\ H^{-\frac{1}{2}}(\partial
D)\rightarrow H^{\frac{1}{2}}(\partial D)$ is not invertible in
$\mathbb{R}^2$. Similar to \cite{AMRZ14}, we introduce a substitute of
$\mathcal {S}^{\alpha,0}$ as follows
\begin{align}\label{3.1}
\widetilde{\mathcal {S}}^{\alpha,0}[\psi]=
\begin{cases}
\mathcal {S}^{\alpha,0}[\phi],\ \ \ \text{if}\ (\psi,\chi(\partial D))_{-\frac{1}{2},\frac{1}{2}}=0,\\
\chi(\partial D),\ \ \ \text{if}\ \psi=\varphi_0,
\end{cases}
\end{align}
where $(\cdot,\cdot)_{-\frac{1}{2},\frac{1}{2}}$ is the duality
pairing between $\ H^{\frac{1}{2}}(\partial D)$ and $\
H^{-\frac{1}{2}}(\partial D)$. $\varphi_0$ is the unique eigenfunction of $(\mathcal
{K}^{-\alpha,0})^*$ associated with eigenvalue $\frac{1}{2}$ such
that $(\varphi_0,\chi(\partial D))_{-\frac{1}{2},\frac{1}{2}}=1$.

From Lemma \ref{lem3.1}, we find that if $(\psi,\chi(\partial
D))_{-\frac{1}{2},\frac{1}{2}}=0$, then $\mathcal
{K}^{-\alpha,0}\widetilde{\mathcal
{S}}^{\alpha,0}[\psi]=\widetilde{\mathcal {S}}^{\alpha,0}(\mathcal
{K}^{-\alpha,0})^*[\psi]$; If $(\psi,\chi(\partial
D))_{-\frac{1}{2},\frac{1}{2}}=1$, we have
\begin{align*}
\widetilde{\mathcal {S}}^{\alpha,0}(\mathcal
{K}^{-\alpha,0})^*[\psi]=\widetilde{\mathcal
{S}}^{\alpha,0}(\mathcal
{K}^{-\alpha,0})^*[\varphi_0]=\frac{1}{2}\widetilde{\mathcal
{S}}^{\alpha,0}[\varphi_0]=\frac{1}{2},
\end{align*}
and
\begin{align*}
\mathcal {K}^{-\alpha,0}\left(\widetilde{\mathcal
{S}}^{\alpha,0}[\psi]\right)=\mathcal {K}^{-\alpha,0}[\chi(\partial
D)]=\frac{1}{2}.
\end{align*}
Thus, we obtain the Calder{\'o}n identity for $\widetilde{\mathcal
{S}}^{\alpha,0}$:
\begin{align}\label{3.2}
\mathcal {K}^{-\alpha,0}\widetilde{\mathcal
{S}}^{\alpha,0}=\widetilde{\mathcal {S}}^{\alpha,0}(\mathcal
{K}^{-\alpha,0})^*.
\end{align}
In view of (\ref{3.2}), we define
\begin{align}\label{3.3}
(u,v)_{\mathcal {H}^*(\partial D)}=-(u,\widetilde{\mathcal
{S}}^{\alpha,0}[v])_{-\frac{1}{2},\frac{1}{2}}.
\end{align}
Thanks to the invertibility and
positivity of $-\widetilde{\mathcal {S}}^{\alpha,0}$, the inner
product (\ref{3.3}) leads to the self-adjointness of $(\mathcal
{K}^{-\alpha,0})^*$ and equivalence between $\mathcal {H}^*(\partial
D)$ and $\ H^{-\frac{1}{2}}(\partial D)$.

\begin{lemma}\label{lem3.2}
Let $D$ be a bounded simply connected domain in $\mathbb{R}^2$. Then

(1) $(\mathcal {K}^{-\alpha,0})^*$ is a compact self-adjoint
operator in the Hilbert space $\mathcal {H}^*(\partial D)$ equipped
with the inner product (\ref{3.3}), which is equivalent to the
original one;

(2) Let $\left\{\lambda_j,\varphi_j\right\},\ j=0,1,2,\cdots$, be the eigenvalue
and normalized eigenfunction pair of $(\mathcal {K}^{-\alpha,0})^*$,
here $\lambda_0=\frac{1}{2}$. Then,
$\lambda_j\in(-\frac{1}{2},\frac{1}{2}]$, and
$\lambda_j\rightarrow0$ as $j\rightarrow\infty$;

(3) $\mathcal {H}^*(\partial D)=\mathcal {H}_0^*(\partial
D)\oplus\{c\varphi_0\},\ c\in\mathbb{C}$, where $\mathcal
{H}_0^*(\partial D)=\{\phi\in\mathcal {H}^*(\partial
D):\int_{\partial D}\phi d\sigma=0\}$;

(4) For any $\psi\in\ H^{-\frac{1}{2}}(\partial D)$, it hold:
\begin{align}\label{3.4}
(\mathcal
{K}^{-\alpha,0})^*[\psi]=\sum_{j=0}^\infty\lambda_j(\psi,\varphi_j)_{\mathcal
{H}^*(\partial D)}\varphi_j.
\end{align}
\end{lemma}

The proof of Lemma \ref{lem3.2} is similar to the classical case (see Lemma C.1 in \cite{AMRZ14}), we omit it here.

Notice that $(\widetilde{\mathcal
{S}}^{\alpha,0})^{-1}[\chi(\partial D)]=\varphi_0$, and
$-\frac{1}{2}\mathcal {I}+(\mathcal
{K}^{-\alpha,0})^*=\left(-\frac{1}{2}\mathcal {I}+(\mathcal
{K}^{-\alpha,0})^*\right)\mathcal {P}_{\mathcal {H}_0^*(\partial
D)}$, where $\mathcal {P}_{\mathcal {H}_0^*(\partial
D)}$ is the orthogonal projection onto $\mathcal {H}_0^*(\partial D)$. Furthermore, it
follows that
\begin{align}\label{3.5}
\left(-\frac{1}{2}\mathcal {I}+(\mathcal
{K}^{-\alpha,0})^*\right)(\widetilde{\mathcal
{S}}^{\alpha,0})^{-1}[\chi(\partial D)]=0.
\end{align}

Owing to (\ref{2.15}), we have
\begin{align}\label{3.6}
\mathcal
{A}_0[\psi]=\sum_{j=0}^\infty\tau_j(\psi,\varphi_j)_{\mathcal
{H}^*(\partial D)}\varphi_j,
\end{align}
where
\begin{align}\label{3.7}
\tau_j=\frac{1}{2}\left(\frac{1}{\mu_m}+\frac{1}{\mu_c}\right)+\left(\frac{1}{\mu_m}-\frac{1}{\mu_c}\right)\lambda_j.
\end{align}

From the spatial representation formula of quasi-periodic Green's
function, we find
\begin{align*}
&G^{\alpha,k}(x,y)=-\frac{i}{4}\sum_{n\in\mathbb{Z}^2}H_0^{(1)}(k|x-n-y|)e^{in\cdot\alpha}\\
=&\sum_{n\in\mathbb{Z}^2}\frac{1}{2\pi}\ln|x-n-y|e^{in\cdot\alpha}+\tau_k\sum_{n\in\mathbb{Z}^2}e^{in\cdot\alpha}+\sum_{j=1}^\infty(k^{2j}\ln
k)b_j\sum_{n\in\mathbb{Z}^2}|x-n-y|^{2j}e^{in\cdot\alpha}\\
&\ \ \ \ \ +\sum_{j=1}^\infty
k^{2j}\sum_{n\in\mathbb{Z}^2}|x-n-y|^{2j}(b_j\ln|x-n-y|+c_j)e^{in\cdot\alpha},
\end{align*}
where $\tau_k=\frac{1}{2\pi}(\ln k+\gamma-\ln2)-\frac{i}{4}$,
$b_j=\frac{(-1)^j}{\pi2^{2j+1}(j!)^2}$,
$c_j=-b_j(\gamma-\ln2-\frac{i\pi}{2}-\sum_{n=1}^j\frac{1}{n})$, and
$\gamma$ is the Euler constant. Thus, we get the following asymptotic expansion formula of quasi-periodic single layer potential
\begin{align}\label{3.8}
\mathcal {S}^{\alpha,k}=\widehat{\mathcal
{S}}^{\alpha,k}+\sum_{j=1}^\infty(k^{2j}\ln k)\mathcal
{S}_j^{\alpha,(1)}+\sum_{j=1}^\infty k^{2j}\mathcal
{S}_j^{\alpha,(2)},
\end{align}
where
\begin{align}\label{3.9}
&\widehat{\mathcal {S}}^{\alpha,k}[\psi](x)=\mathcal
{S}^{\alpha,0}[\psi](x)+\tau_k\sum_{n\in\mathbb{Z}^2}e^{in\cdot\alpha}\int_{\partial
D}\psi(y)d\sigma(y);\\
&\mathcal
{S}_j^{\alpha,(1)}[\psi](x)=\sum_{n\in\mathbb{Z}^2}\int_{\partial
D}b_j|x-n-y|^{2j}e^{in\cdot\alpha}\psi(y)d\sigma(y);\\
&\mathcal
{S}_j^{\alpha,(2)}[\psi](x)=\sum_{n\in\mathbb{Z}^2}\int_{\partial
D}|x-n-y|^{2j}(b_j\ln|x-n-y|+c_j)e^{in\cdot\alpha}\psi(y)d\sigma(y).
\end{align}
Moreover, we can prove (see Lemma C.2 in \cite{AMRZ14}) that, the norms $\|\mathcal {S}_j^{\alpha,(1)}\|_{\mathcal {L}(\mathcal
{H}^*(\partial D),\mathcal {H}(\partial D))}$ and $\|\mathcal
{S}_j^{\alpha,(2)}\|_{\mathcal {L}(\mathcal {H}^*(\partial
D),\mathcal {H}(\partial D))}$ are uniformly bounded with respect to
$j$. Furthermore, the series in (\ref{3.8}) is convergent in space
$\mathcal {L}(\mathcal {H}^*(\partial D),\mathcal {H}(\partial D))$.

\begin{lemma}\label{lem3.3}
As $k$ is small enough, the operator $\mathcal {S}^{\alpha,k}:\
\mathcal {H}^*(\partial D)\rightarrow\mathcal {H}(\partial D)$ is invertible.
\end{lemma}
\begin{proof}
First, we prove that, for $k$ small enough, $\widehat{\mathcal
{S}}^{\alpha,k}:\ \mathcal {H}^*(\partial D)\rightarrow\mathcal
{H}(\partial D)$ is invertible. Since
\begin{align*}
(\mathcal {S}^{\alpha,0}-\widetilde{\mathcal
{S}}^{\alpha,0})[\psi]=&(\mathcal {S}^{\alpha,0}-\widetilde{\mathcal
{S}}^{\alpha,0})[\mathcal {P}_{\mathcal {H}_0^*(\partial
D)}[\psi]+(\psi,\varphi_0)_{\mathcal {H}^*(\partial D)}\varphi_0]\\
=&(\psi,\varphi_0)_{\mathcal {H}^*(\partial D)}(\mathcal
{S}^{\alpha,0}[\varphi_0]-\widetilde{\mathcal
{S}}^{\alpha,0}[\varphi_0])\\
=&(\psi,\varphi_0)_{\mathcal {H}^*(\partial D)}(\mathcal
{S}^{\alpha,0}[\varphi_0]-\chi(\partial D)),
\end{align*}
it follows that
\begin{align}\label{s1}
\widehat{\mathcal {S}}^{\alpha,k}[\psi]=&\widetilde{\mathcal
{S}}^{\alpha,0}[\psi]+(\psi,\varphi_0)_{\mathcal {H}^*(\partial
D)}(\mathcal {S}^{\alpha,0}[\varphi_0]-\chi(\partial
D))\nonumber\\
&+\tau_k\sum_{n\in\mathbb{Z}^2}e^{in\cdot\alpha}\int_{\partial
D}\left(\mathcal {P}_{\mathcal {H}_0^*(\partial
D)}[\psi]+(\psi,\varphi_0)_{\mathcal {H}^*(\partial
D)}\varphi_0\right)d\sigma(y)\nonumber\\
=&\widetilde{\mathcal
{S}}^{\alpha,0}[\psi]+\Upsilon^{\alpha,k}[\psi],
\end{align}
where $$\Upsilon^{\alpha,k}[\psi]=(\psi,\varphi_0)_{\mathcal
{H}^*(\partial D)}\left(\mathcal {S}^{\alpha,0}[\varphi_0]-\chi(\partial
D)+\tau_k\sum_{n\in\mathbb{Z}^2}e^{in\cdot\alpha}\right).$$
Notice that $\widetilde{\mathcal {S}}^{\alpha,0}$ is invertible, then we have $\widehat{\mathcal {S}}^{\alpha,k}(\widetilde{\mathcal
{S}}^{\alpha,0})^{-1}=\mathcal {I}+\Upsilon^{\alpha,k}(\widetilde{\mathcal {S}}^{\alpha,0})^{-1}$. Because of the compactness of $\Upsilon^{\alpha,k}$ and by the Fredholm alternative theorem, we only need to prove the injectivity of $\mathcal {I}+\Upsilon^{\alpha,k}(\widetilde{\mathcal {S}}^{\alpha,0})^{-1}$.

In fact, suppose that $v\in H^{\frac{1}{2}}(\partial D)$ satisfies $(\mathcal {I}+\Upsilon^{\alpha,k}(\widetilde{\mathcal {S}}^{\alpha,0})^{-1})[v]=0$. According to the definition of $\widetilde{\mathcal {S}}^{\alpha,0}$ and $\Upsilon^{\alpha,k}$, if
$(\widetilde{\mathcal {S}}^{\alpha,0})^{-1}[v]\in H_0^{-\frac{1}{2}}(\partial D)$, we have $(\mathcal {I}+\Upsilon^{\alpha,k}(\widetilde{\mathcal {S}}^{\alpha,0})^{-1})[v]=v$, and then $v=0$. If $(\widetilde{\mathcal {S}}^{\alpha,0})^{-1}[v]\in \{\mu\varphi_0,\ \mu\in\mathbb{C}\}$, we see
\begin{align*}
(\mathcal {I}+\Upsilon^{\alpha,k}(\widetilde{\mathcal {S}}^{\alpha,0})^{-1})[v]
=&v+\mu\left(\mathcal {S}^{\alpha,0}[\varphi_0]-\chi(\partial D)+\tau_k\sum_{n\in\mathbb{Z}^2}e^{in\cdot\alpha}\right)\\
=&\mu\left(\mathcal {S}^{\alpha,0}[\varphi_0]+\tau_k\sum_{n\in\mathbb{Z}^2}e^{in\cdot\alpha}\right),
\end{align*}
since we can always find a small enough $k$ such that $\mathcal {S}^{\alpha,0}[\varphi_0]\neq-\tau_k\sum_{n\in\mathbb{Z}^2}e^{in\cdot\alpha}$, it follows that $\mu=0$ and then $v=0$.

Since $\widehat{\mathcal {S}}^{\alpha,k}-\mathcal {S}^{\alpha,k}$ is a compact operator and $\widehat{\mathcal {S}}^{\alpha,k}$
is invertible for $k$ small enough. Furthermore, it is easy to prove that $\mathcal {S}^{\alpha,k}$ is injective for $k$ small enough. In fact, we consider $\psi\in H^{-\frac{1}{2}}(\partial D)$ such that $\mathcal {S}^{\alpha,k}[\psi]=0$. Since $u=\mathcal {S}^{\alpha,k}[\psi]$ satisfies Helmholtz equation $\Delta u+k^2u=0$ in $D$ and $Y\setminus\overline{D}$. Therefore, if $k$ is sufficiently small such that $k^2$ is neither an eigenvalue of $-\Delta$ in
$D$ with the Dirichlet boundary condition on $\partial D$ nor in $Y\setminus\overline{D}$ with the Dirichlet boundary condition on $\partial D$ and the $\alpha$-quasi-periodic condition on $\partial Y$. It follows that $u=0$ and thus, $\psi=\frac{\partial u}{\partial\nu}\big|_+-\frac{\partial u}{\partial\nu}\big|_-=0$, as desired. By using the Fredholm alternative theorem, we see that, as $k$ is small enough, $\mathcal {S}^{\alpha,k}$ is invertible.
\end{proof}

Moreover, under some eigenvalue assumption on wave number $k$, we can also deduce the invertibility of $\mathcal {S}^{\alpha,k}$.
\begin{lemma}\label{lem1}
Suppose that $k^2$ is neither an eigenvalue of $-\Delta$ in
$D$ with the Dirichlet boundary condition on $\partial D$ nor in $Y\setminus
\overline{D}$ with the Dirichlet
boundary condition on $\partial D$ and the $\alpha$-quasi-periodic condition on $\partial Y$. Then
$\mathcal {S}^{\alpha,k}:\ H^{-\frac{1}{2}}(\partial D)\rightarrow H^{\frac{1}{2}}(\partial D)$ is invertible.
Furthermore, $\mathcal {S}^{\alpha,k}$ is an isomorphism from $H^{-\frac{1}{2}}(\partial D)$ to $H^{\frac{1}{2}}(\partial D)$.
\end{lemma}

The proof is similar to Lemma 7.2 in \cite{AKL25} and Theorem 7.3 in \cite{CC26}, so it will be given in Appendix.

\begin{remark}\label{rem3.5}
In fact, as the proof in Lemma \ref{lem3.3}, if $k$ is sufficiently small, then $k$ obviously satisfies the eigenvalue assumption of Lemma \ref{lem1} and the corresponding
results hold. Hence, the invertibility of $\mathcal {S}^{\alpha,k}$ in Lemma \ref{lem3.3} can be regarded as a direct consequence of Lemma \ref{lem1}.
\end{remark}

Clearly, (\ref{3.8}) can be written as
\begin{align}\label{3.12}
\mathcal {S}^{\alpha,k}=\widehat{\mathcal {S}}^{\alpha,k}+\mathcal
{G}^{\alpha,k},
\end{align}
where $\mathcal {G}^{\alpha,k}=(k^{2}\ln k)\mathcal
{S}_1^{\alpha,(1)}+\ k^{2}\mathcal {S}_1^{\alpha,(2)}+\mathcal
{O}(k^{4}\ln k)$. From Lemma \ref{lem3.3}, we get
\begin{align}\label{3.13}
(\mathcal {S}^{\alpha,k})^{-1}=(\mathcal {I}+(\widehat{\mathcal
{S}}^{\alpha,k})^{-1}\mathcal {G}^{\alpha,k})(\widehat{\mathcal
{S}}^{\alpha,k})^{-1},
\end{align}
Noting that $\|(\widehat{\mathcal {S}}^{\alpha,k})^{-1}\|_{\mathcal
{L}(\mathcal {H}(\partial D),\mathcal {H}^*(\partial D))}$ is
bounded for every $k$. Thus, as $k$ is small enough, we have
\begin{align}\label{3.14}
(\mathcal {S}^{\alpha,k})^{-1}=(\widehat{\mathcal
{S}}^{\alpha,k})^{-1}-(\widehat{\mathcal
{S}}^{\alpha,k})^{-1}\mathcal {G}^{\alpha,k}(\widehat{\mathcal
{S}}^{\alpha,k})^{-1}+\mathcal {O}(k^{4}\ln^2 k).
\end{align}
Moreover, setting $\Lambda^{\alpha,k}=(\widetilde{\mathcal
{S}}^{\alpha,0})^{-1}\widehat{\mathcal {S}}^{\alpha,k}$, then
\begin{align*}
\Lambda^{\alpha,k}=&(\widetilde{\mathcal
{S}}^{\alpha,0})^{-1}(\widetilde{\mathcal
{S}}^{\alpha,0}+\Upsilon^{\alpha,k})=(\widetilde{\mathcal
{S}}^{\alpha,0})^{-1}(\widetilde{\mathcal
{S}}^{\alpha,0}+(\cdot,\varphi_0)_{\mathcal {H}^*(\partial
D)}(\mathcal {S}^{\alpha,0}[\varphi_0]-\chi(\partial D)+\tau_k))\\
=&\mathcal {I}+(\cdot,\varphi_0)_{\mathcal {H}^*(\partial
D)}(\mathcal {S}^{\alpha,0}[\varphi_0]-\chi(\partial
D)+\tau_k)\varphi_0.
\end{align*}
And then $$(\Lambda^{\alpha,k})^{-1}=\mathcal
{I}-(\cdot,\varphi_0)_{\mathcal {H}^*(\partial D)}\frac{\mathcal
{S}^{\alpha,0}[\varphi_0]-\chi(\partial D)+\tau_k}{\mathcal
{S}^{\alpha,0}[\varphi_0]+\tau_k}\varphi_0.$$

Hence, we find
\begin{align*}
\left(\widehat{\mathcal
{S}}^{\alpha,k}\right)^{-1}=&(\Lambda^{\alpha,k})^{-1}(\widetilde{\mathcal
{S}}^{\alpha,0})^{-1}=(\widetilde{\mathcal
{S}}^{\alpha,0})^{-1}-((\widetilde{\mathcal
{S}}^{\alpha,0})^{-1}[\cdot],\varphi_0)_{\mathcal {H}^*(\partial
D)}\varphi_0\\
&+\frac{((\widetilde{\mathcal
{S}}^{\alpha,0})^{-1}[\cdot],\varphi_0)_{\mathcal {H}^*(\partial
D)}}{\mathcal {S}^{\alpha,0}[\varphi_0]+\tau_k}\varphi_0,
\end{align*}
then we obtain
\begin{align}\label{3.15}
(\mathcal {S}^{\alpha,k})^{-1}=&\mathcal {L}^\alpha+\mathcal
{U}^{\alpha,k}-(k^{2}\ln k)\mathcal {L}^\alpha\mathcal
{S}_1^{\alpha,(1)}\mathcal {L}^\alpha-k^2(\mathcal
{L}^\alpha\mathcal {S}_1^{\alpha,(2)}\mathcal {L}^\alpha\nonumber\\
&-\ln k(\mathcal {U}^{\alpha,k}\mathcal {S}_1^{\alpha,(1)}\mathcal
{L}^\alpha+\mathcal {L}^\alpha\mathcal {S}_1^{\alpha,(1)}\mathcal
{U}^{\alpha,k}))+\mathcal {O}(k^2\ln^{-1}k),
\end{align}
where $\mathcal {L}^\alpha=\mathcal {P}_{\mathcal {H}_0^*(\partial
D)}(\widetilde{\mathcal {S}}^{\alpha,0})^{-1}$ and $\mathcal
{U}^{\alpha,k}=\frac{((\widetilde{\mathcal
{S}}^{\alpha,0})^{-1}[\cdot],\varphi_0)_{\mathcal {H}^*(\partial
D)}}{\mathcal {S}^{\alpha,0}[\varphi_0]+\tau_k}\varphi_0$. In
particular, $\mathcal {U}^{\alpha,k}=\mathcal {O}(\ln^{-1}k)$.

Now we consider the expansion for operator $(\mathcal
{K}^{-\alpha,0})^*$:
\begin{align}\label{3.16}
(\mathcal {K}^{-\alpha,k})^*=(\mathcal
{K}^{-\alpha,0})^*+\sum_{j=1}^\infty(k^{2j}\ln k)\mathcal
{K}_j^{-\alpha,(1)}+\sum_{j=1}^\infty k^{2j}\mathcal
{K}_j^{-\alpha,(2)},
\end{align}
where
\begin{align*}
&\mathcal
{K}_j^{-\alpha,(1)}[\psi](x)=\sum_{n\in\mathbb{Z}^2}\int_{\partial
D}b_j\frac{\partial|x-n-y|^{2j}}{\partial\nu(x)}e^{-in\cdot\alpha}\psi(y)d\sigma(y);\\
&\mathcal
{K}_j^{-\alpha,(2)}[\psi](x)=\sum_{n\in\mathbb{Z}^2}\int_{\partial
D}\frac{\partial|x-n-y|^{2j}(b_j\ln|x-n-y|+c_j)}{\partial\nu(x)}e^{-in\cdot\alpha}\psi(y)d\sigma(y).
\end{align*}
Likewise, we can prove that the norms $\|\mathcal {K}_j^{-\alpha,(1)}\|_{\mathcal {L}(\mathcal
{H}^*(\partial D),\mathcal {H}^*(\partial D))}$ and $\|\mathcal
{K}_j^{-\alpha,(2)}\|_{\mathcal {L}(\mathcal {H}^*(\partial
D),\mathcal {H}^*(\partial D))}$ are uniformly bounded with respect
to $j$. Furthermore, the series in (\ref{3.15}) is convergent in
space $\mathcal {L}(\mathcal {H}^*(\partial D),\mathcal
{H}^*(\partial D))$.

\begin{lemma}\label{lem3.4}
The operator $\mathcal {A}(\omega):\ \mathcal {H}^*(\partial
D)\rightarrow\mathcal {H}^*(\partial D)$ has the following expansion
formula:
\begin{align}\label{3.18}
\mathcal {A}(\omega)=\mathcal {A}_0+\omega^{2}\ln\omega\mathcal
{A}_1+\mathcal {O}(\omega^2),
\end{align}
where
\begin{align}\label{3.19}
&\mathcal
{A}_0=\frac{1}{2}\left(\frac{1}{\mu_m}+\frac{1}{\mu_c}\right)\mathcal
{I}+\left(\frac{1}{\mu_m}-\frac{1}{\mu_c}\right)(\mathcal
{K}^{-\alpha,0})^*;\\
&\mathcal {A}_1=\frac{1}{\mu_c}\left(\frac{1}{2}\mathcal
{I}-(\mathcal {K}^{-\alpha,0})^*\right)(\widetilde{\mathcal
{S}}^{\alpha,0})^{-1}\mathcal
{S}_1^{\alpha,(1)}\left(\varepsilon_m\mu_m\mathcal
{I}-\varepsilon_c\mu_c\mathcal {P}_{\mathcal {H}_0^*(\partial
D)}\right)\nonumber\\
&\ \ \ \ \ \ \ +\mathcal {K}_1^{-\alpha,(1)}(\varepsilon_m\mathcal
{I}-\varepsilon_c\mathcal {P}_{\mathcal {H}_0^*(\partial D)}).
\end{align}
\end{lemma}
\begin{proof}
From (\ref{3.8}), (\ref{s1}) and (\ref{3.15}), we have that
\begin{align*}
&\mathcal {S}^{\alpha,k_m}=\widetilde{\mathcal
{S}}^{\alpha,0}+\Upsilon^{\alpha,k_m}+(\omega^{2}\ln \omega)\varepsilon_m\mu_m\mathcal
{S}_1^{\alpha,(1)}+\mathcal {O}(\omega^2),\\
&(\mathcal {S}^{\alpha,k_c})^{-1}=\mathcal {L}^\alpha+\mathcal
{U}^{\alpha,k_c}-(\omega^{2}\ln \omega)\varepsilon_c\mu_c\mathcal {L}^\alpha\mathcal
{S}_1^{\alpha,(1)}\mathcal {L}^\alpha+\mathcal {O}(\omega^2).
\end{align*}
Noting the definition of $\Upsilon^{\alpha,k_m}$ in (\ref{s1}), we see $\mathcal {L}^\alpha\Upsilon^{\alpha,k_m}=0$,
and then
\begin{align*}
(\mathcal {S}^{\alpha,k_c})^{-1}\mathcal {S}^{\alpha,k_m}=&\mathcal {P}_{\mathcal {H}_0^*(\partial
D)}+\mathcal {U}^{\alpha,k_c}\widetilde{\mathcal {S}}^{\alpha,0}+\mathcal {U}^{\alpha,k_c}\Upsilon^{\alpha,k_m}\\
&\ \ \ -(\omega^{2}\ln \omega)\mathcal {L}^\alpha\mathcal
{S}_1^{\alpha,(1)}\left(\varepsilon_m\mu_m\mathcal {I}-\varepsilon_c\mu_c\mathcal {P}_{\mathcal {H}_0^*(\partial
D)}\right)+\mathcal {O}(\omega^2).
\end{align*}
By using (\ref{3.5}), it yields $\left(-\frac{1}{2}\mathcal {I}+(\mathcal
{K}^{-\alpha,0})^*\right)\mathcal {U}^{\alpha,k_c}=0$. Since $-\frac{1}{2}\mathcal {I}+(\mathcal
{K}^{-\alpha,k})^*=\left(-\frac{1}{2}\mathcal {I}+(\mathcal
{K}^{-\alpha,0})^*\right)-(k^{2}\ln k)\mathcal
{K}_1^{-\alpha,(1)}+\mathcal {O}(k^2)$, we obtain the result.
\end{proof}

\begin{theorem}\label{thm3.5}
Let $\left\{\tau_j(\omega),\varphi_j(\omega)\right\},\ j=0,1,2,\cdots$, be the eigenvalue
and normalized eigenfunction pair of $\mathcal {A}(\omega)$, then, in the quasi-static regime, we have
the following expansion formula
\begin{align}\label{3.21}
&\tau_j(\omega)=\tau_j+(\omega^{2}\ln\omega)\tau_{j,1}+\mathcal
{O}(\omega^2);\\
&\varphi_j(\omega)=\varphi_j+(\omega^{2}\ln\omega)\varphi_{j,1}+\mathcal
{O}(\omega^2),\label{3.22}
\end{align}
where $\tau_j$ and $\varphi_j$ are defined by (\ref{3.7}) and in Lemma \ref{lem3.2} respectively,
\begin{align}
\tau_{j,1}=&R_{jj},\label{16}\\
\varphi_{j,1}=&\sum_{j\neq
l}\frac{R_{jl}}{\left(\frac{1}{\mu_m}-\frac{1}{\mu_c}\right)(\lambda_j-\lambda_l)}\varphi_l,\\
R_{j,l}=&(\mathcal {A}_1[\varphi_j],\varphi_l)_{\mathcal
{H}^*(\partial D)}.
\end{align}
\end{theorem}
\begin{proof}
We substitute (\ref{3.18}), (\ref{3.21}) and (\ref{3.22}) in the characteristic equation $\mathcal {A}(\omega)[\varphi_j(\omega)]=\tau_j(\omega)\varphi_j(\omega)$, and compare the coefficients of $\omega$, $\omega^2$ respectively, then the results hold.
\end{proof}

Similar to \cite{AMRZ14}, we give the definition for index set of plasmonic
resonance and some mild conditions.

\begin{definition}\label{de2}
We say that $J\subseteq \mathbb{N}$ is an index set of
resonance if $\tau_j$ are close to zero when $j\in J$ and are
bounded from below when $j\in J^c$. More precisely, we choose a
threshold number $\eta_0>0$ independent of $\omega$ such that
$|\tau_j|\geq\eta_0>0$, for $j\in J^c$.
\end{definition}

Next, we assume that the following conditions are satisfied:

(C 1) The electric permittivity and magnetic permeability
$\varepsilon_m,\ \varepsilon_c,\ \mu_m,\ \mu_c$ are dimensionless
and are of order one, the particle $D$ has size of order one,
$\omega$ is dimensionless and is of order $o(1)$.

(C 2) Each eigenvalue $\lambda_j$ for $j\in J$ is a simple
eigenvalue of the operator $(\mathcal {K}^{-\alpha,0})^*$.

(C 3) Let
\begin{align}\label{3.26}
\lambda(t)=\frac{1+t}{2(1-t)}.
\end{align}
We suppose
that $\lambda\left(\frac{\mu_c}{\mu_m}\right)\neq0$, i.e. $\frac{\mu_c}{\mu_m}\neq-1$.

Notice that for $j=0$, since $\lambda_0=\frac{1}{2}$ and (\ref{3.7}), we see
$\tau_0=\frac{1}{\mu_m}$, which is of size one by our assumption.
Thus, throughout this paper, we always exclude $0$ from the index
set $J$.

We define the projection operator $P_J(\omega)$ as follows
\begin{align}\label{}
P_J(\omega)[\varphi_j(\omega)]=
\begin{cases}
\varphi_j(\omega),\ \ \ j\in J,\\
0,\ \ \ j\in J^c.
\end{cases}
\end{align}

Let $\{\tilde{\tau}_j(\omega); \tilde{\varphi}_j(\omega)\}$ be the eigenvalue system of $\mathcal {A}^*(\omega)$, we can get the
similar expansion formula
\begin{align}\label{}
&\tilde{\tau}_j(\omega)=\overline{\tau_j(\omega)};\\
&\tilde{\varphi}_j(\omega)=\varphi_j+(\omega^{2}\ln\omega)\tilde{\varphi}_{j,1}+\mathcal
{O}(\omega^2).\label{3.28}
\end{align}
Applying the eigenfunction expansion, we have
\begin{align}\label{}
P_J(\omega)[x]=\sum_{j\in J}\left(x,\tilde{\varphi}_j(\omega)\right)_{\mathcal{H}^*(\partial D)}\varphi_j(\omega),\ \ \ x\in \mathcal{H}^*(\partial D).
\end{align}

Hence, from (\ref{2.12}), we obtain
\begin{align}\label{3.30}
\psi=\mathcal {A}^{-1}(\omega)[f]=\sum_{j\in J}\frac{\left(f,\tilde{\varphi}_j(\omega)\right)_{\mathcal{H}^*(\partial D)}}{\tau_j(\omega)}\varphi_j(\omega)+\mathcal {A}^{-1}(\omega)[P_{J^c}(\omega)[f]],
\end{align}
and it can be proved that, as $\omega$ is small enough, $\|\mathcal {A}^{-1}(\omega)P_{J^c}(\omega)\|\lesssim 1+\mathcal{O}(\omega)$ (see Lemma 2.5 in \cite{AMRZ14}). Here and throughout this paper, $A\lesssim B$ means $A\leq CB$ for some positive constant $C$ independent of parameters involved. $A\approx B$ means that $A\lesssim B$ and $B\lesssim A$.

\begin{theorem}\label{thm3.8}
In the quasi-static regime, the near field $u$ has
the following representation
\begin{align}\label{3.32}
u(x)=
\begin{cases}
\mathcal {S}^{\alpha,k_c}[\phi](x),\ \ \ x\in D,\\
F_z(x)+\mathcal {S}^{\alpha,k_m}[\psi](x),\ \ \ x\in
Y\setminus\overline{D},
\end{cases}
\end{align}
where $F_z$ is defined by (\ref{2.9}), $(\phi,\psi)\in H^{-\frac{1}{2}}(\partial D)\times H^{-\frac{1}{2}}(\partial D)$ is given by
\begin{align}\label{}
&\phi=\left(\mathcal {S}^{\alpha,k_c}\right)^{-1}\mathcal {S}^{\alpha,k_m}[\psi],\\
&\psi=\sum_{j\in J}\frac{\omega\left(f_1,\varphi_j\right)_{\mathcal{H}^*(\partial D)}+\mathcal
{O}(\omega^3\ln\omega)}{\lambda\left(\frac{\mu_c}{\mu_m}\right)-\lambda_j+\mathcal
{O}(\omega^2\ln\omega)}\varphi_j+\mathcal{O}(\omega)
\end{align}
with $\lambda(t)$ being defined by (\ref{3.26}), and
\begin{align}\label{3.35}
f_1=&-\frac{i\mu_c\sqrt{\varepsilon_m\mu_m}}{4(\mu_c-\mu_m)}\bigg(\sum_{n\in\mathbb{Z}^2}\bigg(aH_1^{(1)}(k_m|x-n-z|)\nonumber\\
&-a\cdot(x-n-z)k_mH_2^{(1)}(k_m|x-n-z|)\frac{(x-n-z)}{|x-n-z|}\bigg)\frac{e^{in\cdot\alpha}}{|x-n-z|},\nu(x)\bigg),
\end{align}
here $H_1^{(1)}$, $H_2^{(1)}$ denote the Hankel functions of the first of order 1 and the Hankel functions of the second of order 1 respectively.
\end{theorem}
\begin{proof}
According to the representation formula of quasi-periodic Green's function, it is easy to find that
\begin{align}\label{}
f=-\frac{1}{\mu_m}\frac{\partial F_z}{\partial\nu}=\omega f_1,
\end{align}
where $f_1$ is defined by (\ref{3.35}).
Then, by (\ref{3.7}), (\ref{3.28}), (\ref{3.30}) and Theorem \ref{thm3.5},
we find
\begin{align}\label{}
\psi=\mathcal {A}^{-1}(\omega)[f]=\sum_{j\in J}\frac{\omega\left(f_1,\varphi_j\right)_{\mathcal{H}^*(\partial D)}+\mathcal
{O}(\omega^3\ln\omega)}{\lambda\left(\frac{\mu_c}{\mu_m}\right)-\lambda_j+\mathcal
{O}(\omega^2\ln\omega)}\varphi_j+\mathcal{O}(\omega).
\end{align}
Thus, from (\ref{2.11}) and solution formula (\ref{2.8}), it yields the result.
\end{proof}
\begin{remark}
According to Theorem \ref{thm3.8}, we can see that if $\lambda\left(\frac{\mu_c}{\mu_m}\right)\rightarrow\lambda_j$ or $\tau_j\rightarrow0$, for some $j\in J$, the near field $u(x), x\in D$
blows up, i.e., plasmonic resonance occurs. Furthermore, in next section, we will analysis the rate of blow up of near field energy.
\end{remark}

\section{Analysis for the rate of blow up of near field energy}

In order to estimate the rate of blow up in near field energy, we first give the estimation for gradient of the solution.

\begin{lemma}\label{lem4.1}
Let $\phi=\tilde{\phi}+(\phi,\varphi_0)\varphi_0$, here $\tilde{\phi}\in\mathcal{H}_0^*(\partial D)$, for the solution
of (\ref{2.3}) in $D$, i.e., $u=\mathcal {S}^{\alpha,k_c}[\phi]$, we have the following estimation
\begin{align}
\left|\|\nabla u\|_{L^2(D)}^2-\|\tilde{\phi}\|_{\mathcal{H}^*(\partial D)}^2\right|\lesssim\left|\omega\ln\omega\right|^2\left|(\phi,\varphi_0)\right|^2.
\end{align}
\end{lemma}

\begin{proof}
Using divergence theorem in $D$, we have that
\begin{align*}
\int_D|\nabla u|^2dx\leq |k_c|^2\int_D|u|^2dx+\left|\int_{\partial D}u\overline{\frac{\partial
u}{\partial\nu}\Big|_-}d\sigma\right|.
\end{align*}
By (\ref{3.8}) and solution formula (\ref{2.8}), it deduces
\begin{align*}
|u|^2=\left|\mathcal {S}^{\alpha,k_c}[\phi]\right|^2\lesssim\left|\widehat{\mathcal
{S}}^{\alpha,k_c}[\phi]\right|^2+\left|\sum_{j=1}^\infty(k_c^{2j}\ln k_c)\mathcal
{S}_j^{\alpha,(1)}[\phi]\right|^2+\left|\sum_{j=1}^\infty k_c^{2j}\mathcal
{S}_j^{\alpha,(2)}[\phi]\right|^2.
\end{align*}
Since $\tau_{k_c}\lesssim \ln k_c$ and $\omega$ is small enough, we get
\begin{align}\label{4.2}
\|u\|_{L^2(D)}^2\lesssim |\ln k_c|\|\phi\|_{\mathcal{H}^*(\partial D)}^2.
\end{align}
Noting the Poisson summation formula
\begin{align*}
\sum_{n\in\mathbb{Z}^2}e^{in\cdot\alpha}=\sum_{n\in\mathbb{Z}^2}\delta(\frac{\alpha}{2\pi}-n),
\end{align*}
from (\ref{3.8}), (\ref{3.16}) and the jump formula, for every $\alpha\in(0,2\pi)^2$, we get
\begin{align}
\left|\int_{\partial D}u\overline{\frac{\partial
u}{\partial\nu}\Big|_-}d\sigma\right|=&\left|\int_{\partial D}\mathcal {S}^{\alpha,k_c}[\phi]\overline{\frac{\partial
}{\partial\nu}\mathcal {S}^{\alpha,k_c}[\phi]\Big|_-}d\sigma\right|\nonumber\\
=&\left|\int_{\partial D}\mathcal {S}^{\alpha,k_c}[\phi]\overline{\left(-\frac{1}{2}\mathcal
{I}+(\mathcal {K}^{-\alpha,k_c})^*\right)[\phi]}d\sigma\right|\nonumber\\
\leq&\left|\int_{\partial D}\mathcal {S}^{\alpha,0}[\phi]\overline{\left(-\frac{1}{2}\mathcal
{I}+(\mathcal{K}^{-\alpha,0})^*\right)[\phi]}d\sigma\right|+\left|E\right|\nonumber,
\end{align}
where
\begin{align}
&E=\int_{\partial D}\mathcal{S}^{\alpha,0}[\phi]\overline{\sum_{j=1}^\infty(k_c^{2j}\ln k_c)\mathcal{K}_j^{-\alpha,(1)}[\phi]}d\sigma
+\int_{\partial D}\mathcal{S}^{\alpha,0}[\phi]\overline{\sum_{j=1}^\infty k_c^{2j}\mathcal{K}_j^{-\alpha,(2)}[\phi]}d\sigma\nonumber\\
&+\int_{\partial D}\sum_{j=1}^\infty(k_c^{2j}\ln k_c)\mathcal
{S}_j^{\alpha,(1)}[\phi]\overline{\left(-\frac{1}{2}\mathcal
{I}+(\mathcal {K}^{-\alpha,k_c})^*\right)[\phi]}d\sigma\nonumber\\
&+\int_{\partial D}\sum_{j=1}^\infty k_c^{2j}\mathcal
{S}_j^{\alpha,(2)}[\phi]\overline{\left(-\frac{1}{2}\mathcal
{I}+(\mathcal {K}^{-\alpha,k_c})^*\right)[\phi]}d\sigma\nonumber.
\end{align}
For $k_c$ small enough, it is easy to find that $|E|\lesssim\left|k_c\ln k_c\right|^2\|\phi\|_{\mathcal{H}^*(\partial D)}^2$. Next, we estimate
$\left|\int_{\partial D}\mathcal {S}^{\alpha,0}[\phi]\overline{\left(-\frac{1}{2}\mathcal
{I}+(\mathcal{K}^{-\alpha,0})^*\right)[\phi]}d\sigma\right|$.

Since $(\mathcal {K}^{-\alpha,k_c})^*[\varphi_0]=\frac{1}{2}\varphi_0$ and $\mathcal{S}^{\alpha,0}[\varphi_0]=0$, for $\phi=\tilde{\phi}+(\phi,\varphi_0)\varphi_0$, we see that
\begin{align*}
&\int_{\partial D}\mathcal {S}^{\alpha,0}[\phi]\overline{\left(-\frac{1}{2}\mathcal
{I}+(\mathcal{K}^{-\alpha,0})^*\right)[\phi]}d\sigma=\int_{\partial D}\mathcal {S}^{\alpha,0}[\tilde{\phi}]\overline{\left(-\frac{1}{2}\mathcal
{I}+(\mathcal{K}^{-\alpha,0})^*\right)[\tilde{\phi}]}d\sigma\\
=&\int_{\partial D}\mathcal {S}^{\alpha,0}\left[\sum_{j=1}^\infty(\tilde{\phi},\varphi_j)\varphi_j\right]\overline{\left(-\frac{1}{2}\mathcal
{I}+(\mathcal{K}^{-\alpha,0})^*\right)\left[\sum_{l=1}^\infty(\tilde{\phi},\varphi_l)\varphi_l\right]}d\sigma\\
=&\int_{\partial D}\sum_{j=1}^\infty(\tilde{\phi},\varphi_j)\mathcal{S}^{\alpha,0}[\varphi_j]\overline{\sum_{l=1}^\infty(\tilde{\phi},\varphi_l)
\left(-\frac{1}{2}+\lambda_l\right)\varphi_l}d\sigma\\
=&\sum_{j,l=1}^\infty\left(-\frac{1}{2}+\lambda_l\right)\overline{(\tilde{\phi},\varphi_l)}(\tilde{\phi},\varphi_j)\int_{\partial D}\mathcal{S}^{\alpha,0}[\varphi_j]\overline{\varphi_l}d\sigma,
\end{align*}
Notice that $\int_{\partial D}\mathcal{S}^{\alpha,0}[\varphi_j]\overline{\varphi_l}d\sigma=-\left(\varphi_l,\varphi_j\right)_{\mathcal{H}^*(\partial D)}=-\delta_{lj}$ (the Kronecker's delta), we have that
\begin{align*}
\int_{\partial D}\mathcal {S}^{\alpha,0}[\phi]\overline{\left(-\frac{1}{2}\mathcal
{I}+(\mathcal{K}^{-\alpha,0})^*\right)[\phi]}d\sigma
=\sum_{j=1}^\infty\left(-\frac{1}{2}+\lambda_j\right)\left|(\tilde{\phi},\varphi_j)\right|^2.
\end{align*}
Since $\lambda_j\in\left(-\frac{1}{2},\frac{1}{2}\right)$ $(j\geq1)$, it implies
\begin{align}\label{4.3}
\left|\int_{\partial D}\mathcal {S}^{\alpha,0}[\phi]\overline{\left(-\frac{1}{2}\mathcal
{I}+(\mathcal{K}^{-\alpha,0})^*\right)[\phi]}d\sigma
\right|\approx\|\tilde{\phi}\|_{\mathcal{H}^*(\partial D)}^2.
\end{align}
Combing (\ref{4.2}), (\ref{4.3}) and the estimation of $E$, we obtain the result.
\end{proof}

For simplicity, let $\sigma=\Re\left(\frac{1}{\mu_c(\omega)}\right)$, $\delta=\Im\left(\frac{1}{\mu_c(\omega)}\right)<0$, then
$\tau_j$ given by (\ref{3.7}) can be written as
\begin{align}\label{4.4}
\tau_j=\frac{1}{2}(\sigma+\mu_m^{-1})-(\sigma-\mu_m^{-1})\lambda_j+\delta(\frac{1}{2}-\lambda_j)i.
\end{align}
Next, we define the approximate
density $\psi_0$ as the solution of equation $\mathcal{A}_0[\psi_0]=f$, here $\mathcal{A}_0$ is given by (\ref{3.19}), then,
similar to (\ref{3.30}), applying the eigenfunction expansion, it follows that
\begin{align}\label{4.5}
\psi_0=\mathcal {A}_0^{-1}[f]=\sum_{j=0}^\infty\frac{\left(f,\varphi_j\right)_{\mathcal{H}^*(\partial D)}}{\tau_j}\varphi_j.
\end{align}
\begin{lemma}\label{lem4.2}
Under Conditions (C1), (C2) and (C3), $\psi_0$ is given by (\ref{4.5}) and has the decomposition $\psi_0=\tilde{\psi}_0+c\varphi_0$, ($\tilde{\psi}_0\in\mathcal{H}^*(\partial D)$,
$c$ is a constant). Moreover, let $\sigma=\Re\left(\mu_c^{-1}(\omega)\right)$, $\delta=\Im\left(\mu_c^{-1}(\omega)\right)$. Then, for sufficiently small $|\delta|$, we have that

(1) $\|\mathcal{A}_0^{-1}\|_{\mathcal{L}(\mathcal{H}^*(\partial D),\mathcal{H}^*(\partial D))}\lesssim\delta^{-1}$.

(2) If $\lambda(\sigma^{-1}\mu_m^{-1})\neq\lambda_j, (j\geq0)$, then $\|\mathcal{A}_0^{-1}\|_{\mathcal{L}(\mathcal{H}^*(\partial D),\mathcal{H}^*(\partial D))}\lesssim C$ for some positive constant $C$.

(3) If $\lambda(\sigma^{-1}\mu_m^{-1})=\lambda_j$ for some $j\geq1$, then $\|\tilde{\psi}_0\|_{\mathcal{H}^*(\partial D)}\gtrsim|\delta|^{-1}\left|\left(f,\varphi_j\right)_{\mathcal{H}^*(\partial D)}\right|$.
\end{lemma}

\begin{proof}
(1) Since, for $j\neq0$, $\left|\tau_j^{-1}\right|\lesssim\frac{1}{|\delta|(\frac{1}{2}-\lambda_j)}\lesssim|\delta|^{-1}$, then it find
\begin{align}\label{}
\|\psi_0\|_{\mathcal{H}^*(\partial D)}^2\lesssim|\delta|^{-2}\sum_{j=0}^\infty\left|\left(f,\varphi_j\right)_{\mathcal{H}^*(\partial D)}\right|^2\lesssim|\delta|^{-2}\|f\|_{\mathcal{H}^*(\partial D)}^2.
\end{align}
Thus, $\|\mathcal{A}_0^{-1}\|_{\mathcal{L}(\mathcal{H}^*(\partial D),\mathcal{H}^*(\partial D))}\lesssim|\delta|^{-1}$.

(2) If $\lambda(\sigma^{-1}\mu_m^{-1})\neq\lambda_j$, then, for $j\geq0$, we find $\left|\lambda(\sigma\mu_m^{-1})-\lambda_j\right|\geq c_0$, here $c_0$
is a positive constant. Hence, $\left|\tau_j^{-1}\right|\lesssim1$ and $\|\psi_0\|_{\mathcal{H}^*(\partial D)}^2\lesssim\|f\|_{\mathcal{H}^*(\partial D)}^2$, i.e.,
$\|\mathcal{A}_0^{-1}\|_{\mathcal{L}(\mathcal{H}^*(\partial D),\mathcal{H}^*(\partial D))}\lesssim C$.

(3) If $\lambda(\sigma^{-1}\mu_m^{-1})=\lambda_j$ for some $j\geq1$, then
\begin{align}\label{}
\|\tilde{\psi}_0\|_{\mathcal{H}^*(\partial D)}\gtrsim\left|\left(\tilde{\psi}_0,\varphi_j\right)_{\mathcal{H}^*(\partial D)}\right|\gtrsim\frac{\left|\left(f,\varphi_j\right)_{\mathcal{H}^*(\partial D)}\right|}{\left|\delta(\frac{1}{2}-\lambda_j)\right|}\gtrsim|\delta|^{-1}\left|\left(f,\varphi_j\right)_{\mathcal{H}^*(\partial D)}\right|.
\end{align}
\end{proof}

\begin{theorem}\label{thm4.3}
Let $u$ be the solution of (\ref{2.3}), and $\sigma=\Re\left(\mu_c^{-1}(\omega)\right)$, $\delta=\Im\left(\mu_c^{-1}(\omega)\right)$.
Suppose that $|\delta|^{-1}\omega\leq c_1$ for sufficiently small $c_1$, we have

(1) If $\lambda(\sigma^{-1}\mu_m^{-1})\neq\lambda_j$ for any $j\geq0$, then there exist a constant $C$ independent of $\delta$ such that
\begin{align}
\left\|\nabla u\right\|_{L^2(D)}\leq C.
\end{align}

(2) If $\lambda(\sigma^{-1}\mu_m^{-1})=\lambda_j$ for some $j\geq1$, let $z$ be such that $a\cdot\nabla_z\mathcal {S}^{\alpha,0}[\varphi_j](z)\neq0$.
Then
\begin{align}\label{4.9}
\left\|\nabla u\right\|_{L^2(D)}\approx|\delta|^{-1}
\end{align}
as $|\delta|$ is small enough.
\end{theorem}
\begin{proof}
(1) From (\ref{3.18}), it follows that
\begin{align}\label{}
\mathcal {A}(\omega)=\mathcal{A}_0\left(\mathcal{I}+\mathcal{A}_0^{-1}\left(\omega^{2}\ln\omega\mathcal
{A}_1+\mathcal {O}(\omega^2)\right)\right),
\end{align}
and then
\begin{align}\label{}
\psi=\left(\mathcal{I}+\mathcal{A}_0^{-1}\left(\omega^{2}\ln\omega\mathcal
{A}_1+\mathcal {O}(\omega^2)\right)\right)^{-1}\mathcal{A}_0^{-1}[f].
\end{align}
By using Lemma \ref{lem4.2} (1), we see
\begin{align}\label{4.10}
\left\|\mathcal{A}_0^{-1}\left(\omega^{2}\ln\omega\mathcal
{A}_1+\mathcal {O}(\omega^2)\right)\right\|_{\mathcal{L}(\mathcal{H}^*(\partial D),\mathcal{H}^*(\partial D))}\lesssim|\delta|^{-1}\omega.
\end{align}
Therefore,
\begin{align}
\left\|\psi-\psi_0\right\|_{\mathcal{H}^*(\partial D)}=&\left\|\left(\mathcal{I}+\mathcal{A}_0^{-1}\left(\omega^{2}\ln\omega\mathcal
{A}_1+\mathcal {O}(\omega^2)\right)\right)^{-1}\mathcal{A}_0^{-1}[f]-\mathcal{A}_0^{-1}[f]\right\|_{\mathcal{H}^*(\partial D)}\nonumber\\
\lesssim&|\delta|^{-1}\omega\left\|\mathcal{A}_0^{-1}[f]\right\|_{\mathcal{H}^*(\partial D)}\nonumber\\
=&|\delta|^{-1}\omega\left\|\psi_0\right\|_{\mathcal{H}^*(\partial D)}.\label{4.11}
\end{align}
If $\lambda(\sigma^{-1}\mu_m^{-1})\neq\lambda_j$, by Lemma \ref{lem4.2} (2), we get
\begin{align*}
\left\|\psi\right\|_{\mathcal{H}^*(\partial D)}\lesssim(1+|\delta|^{-1}\omega)\left\|\psi_0\right\|_{\mathcal{H}^*(\partial D)}
=(1+|\delta|^{-1}\omega)\left\|\mathcal{A}_0^{-1}[f]\right\|_{\mathcal{H}^*(\partial D)}\lesssim\left\|f\right\|_{\mathcal{H}^*(\partial D)}.
\end{align*}
Then, by Lemma \ref{lem4.1}, Lemma \ref{lem1} and Remark \ref{rem3.5}, it yields
\begin{align*}
\left\|\nabla u\right\|_{L^2(D)}=\left\|\nabla\mathcal{S}^{\alpha,k_c}[\phi]\right\|_{\mathcal{H}(\partial D)}\lesssim\left\|\phi\right\|_{\mathcal{H}^*(\partial D)}\lesssim\left\|\psi\right\|_{\mathcal{H}^*(\partial D)}\lesssim\left\|f\right\|_{\mathcal{H}^*(\partial D)}.
\end{align*}

(2) If $\lambda(\sigma^{-1}\mu_m^{-1})=\lambda_j$, ($j\geq1$), then, from Lemma \ref{lem4.2} (3), we find
\begin{align}\label{4.12}
\|\tilde{\psi}_0\|_{\mathcal{H}^*(\partial D)}\gtrsim|\delta|^{-1}\left|\left(f,\varphi_j\right)_{\mathcal{H}^*(\partial D)}\right|.
\end{align}
Moreover, by (\ref{4.11}), it deduces
\begin{align}\label{4.13}
|\delta|^{-1}\omega\left\|\psi_0\right\|_{\mathcal{H}^*(\partial D)}\gtrsim\left\|\psi-\psi_0\right\|_{\mathcal{H}^*(\partial D)}
\gtrsim\left\|\tilde{\psi}_0\right\|_{\mathcal{H}^*(\partial D)}-\left\|\tilde{\psi}\right\|_{\mathcal{H}^*(\partial D)},
\end{align}
where $\tilde{\psi}=\psi-(\psi,\varphi_0)_{\mathcal{H}^*(\partial D)}\varphi_0$.

Combining now (\ref{4.12}) and (\ref{4.13}), and noting $|\delta|^{-1}\omega\leq c_1$,  for sufficiently small $|\delta|$, we have that
\begin{align}\label{}
\left\|\tilde{\psi}\right\|_{\mathcal{H}^*(\partial D)}\gtrsim\left\|\tilde{\psi}_0\right\|_{\mathcal{H}^*(\partial D)}-|\delta|^{-1}\omega\left\|\psi_0\right\|_{\mathcal{H}^*(\partial D)}\gtrsim|\delta|^{-1}\left|\left(f,\varphi_j\right)_{\mathcal{H}^*(\partial D)}\right|.
\end{align}
Hence, we obtain from Lemma \ref{lem4.1} and Remark \ref{rem3.5} that
\begin{align*}
\left\|\nabla u\right\|_{L^2(D)}^2=&\left\|\nabla\mathcal{S}^{\alpha,k_c}[\phi]\right\|_{\mathcal{H}(\partial D)}^2\gtrsim
\left\|\tilde{\phi}\right\|_{\mathcal{H}^*(\partial D)}^2-\left|\omega\ln\omega\right|^2\left|(\phi,\varphi_0)_{\mathcal{H}^*(\partial D)}\right|^2\\
\gtrsim&\left\|\tilde{\psi}\right\|_{\mathcal{H}^*(\partial D)}^2-\left|\omega\ln\omega\right|^2\left|(\psi,\varphi_0)_{\mathcal{H}^*(\partial D)}\right|^2\\
\gtrsim&|\delta|^{-2}\left|(f,\varphi_j)_{\mathcal{H}^*(\partial D)}\right|^2-\left|\omega\ln\omega\right|^2\left|(\psi,\varphi_0)_{\mathcal{H}^*(\partial D)}\right|^2.
\end{align*}
Next, we will show that $\left|(\psi,\varphi_0)_{\mathcal{H}^*(\partial D)}\right|$ is bounded, and $(f,\varphi_j)_{\mathcal{H}^*(\partial D)}\neq0$.

Let $\left(\mathcal{I}+\left(\omega^{2}\ln\omega\mathcal
{A}_1+\mathcal {O}(\omega^2)\right)\mathcal{A}_0^{-1}\right)^{-1}[f]=g$. Since $|\delta|^{-1}\omega\leq c_1$, by (\ref{4.10}), it implies
$\left\|g\right\|_{\mathcal{H}^*(\partial D)}$ is bounded. Then, noting $\lambda_0=\frac{1}{2}$, we see that
\begin{align*}
\left|(\psi,\varphi_0)_{\mathcal{H}^*(\partial D)}\right|=\left|\tau_0^{-1}(g,\varphi_0)_{\mathcal{H}^*(\partial D)}\right|=\left|\mu_m(g,\varphi_0)_{\mathcal{H}^*(\partial D)}\right|\lesssim C.
\end{align*}
Here, $C$ is a positive constant.

Applying Green's formula and jump relation, it follows that
\begin{align*}
&(f,\varphi_j)_{\mathcal{H}^*(\partial D)}=-\left(\frac{\partial F_z}{\partial\nu},\varphi_j\right)_{\mathcal{H}^*(\partial D)}
=\left(\frac{\partial F_z}{\partial\nu},\mathcal {\tilde{S}}^{\alpha,0}[\varphi_j]\right)\\
=&\int_{D}\mathcal {S}^{\alpha,0}[\varphi_j]\Delta F_zdx-\int_{D}F_z\Delta\mathcal {S}^{\alpha,0}[\varphi_j]dx
+\int_{\partial D}F_z\frac{\partial\mathcal{S}^{\alpha,0}}{\partial\nu}[\varphi_j]\Big|_-d\sigma\\
=&-k_m^2\int_{D}F_z\mathcal {S}^{\alpha,0}[\varphi_j]dx+\left(\lambda_j-\frac{1}{2}\right)\int_{\partial D}F_z\varphi_jd\sigma.
\end{align*}
Furthermore, from (\ref{3.8}) and (\ref{2.9}), it yields
\begin{align*}
&\int_{\partial D}F_z\varphi_jd\sigma=\int_{\partial D}a\cdot\nabla_x G^{\alpha,k_m}(x,z)\varphi_jd\sigma
=-\int_{\partial D}a\cdot\nabla_z G^{\alpha,k_m}(x,z)\varphi_jd\sigma\\
=&-a\cdot\nabla_z\mathcal {S}^{\alpha,k_m}[\varphi_j](z)=-a\cdot\nabla_z\mathcal {S}^{\alpha,0}[\varphi_j](z)+\mathcal{O}\left(\omega^{2}\ln\omega\right).
\end{align*}
Thereby, $(f,\varphi_j)_{\mathcal{H}^*(\partial D)}=-a\left(\lambda_j-\frac{1}{2}\right)\cdot\nabla_z\mathcal {S}^{\alpha,0}[\varphi_j](z)+\mathcal{O}\left(\omega^{2}\ln\omega\right)$. Clearly, for $\omega$ small enough,
$(f,\varphi_j)_{\mathcal{H}^*(\partial D)}\neq0$. We find $\left\|\nabla u\right\|_{L^2(D)}\gtrsim|\delta|^{-1}$, and likewise $|\delta|^{-1}\gtrsim\left\|\nabla u\right\|_{L^2(D)}$ can also be proved. Hence, (\ref{4.9}) hold.
\end{proof}

\begin{remark}
From the proof of Theorem \ref{thm4.3} (2), it shows that If $\lambda(\sigma^{-1}\mu_m^{-1})=\lambda_j$ for some $j\geq1$, in the quasi-static regime and for fixed parameter $\delta$,  the near energy
\begin{align}\label{4.9}
\left\|\nabla u\right\|_{L^2(D)}\gtrsim \left|\frac{a}{\delta}\left(\lambda_j-\frac{1}{2}\right)\cdot\nabla_z\mathcal {S}^{\alpha,0}[\varphi_j](z)\right|.
\end{align}
Then, because of the singularity of $\nabla_z\mathcal {S}^{\alpha,0}[\varphi_j](z)$, we find that when the the location of the dipole source $z\rightarrow \partial D$ or the intensity of the dipole source $|a|$ is large enough, the near energy will become large, which indicates the enhancement of near field.
\end{remark}

Now we write $\mu_c(\omega)=\mu^{(1)}(\omega)+i\mu^{(2)}(\omega)$,
and $\mu^{(1)},\ \mu^{(2)}$ satisfy the following Kramer-Kronig
relations (Hilbert transform) \cite{AMRZ14,ADM6}:
\begin{align}\label{kk}
\mu^{(1)}(\omega)=&\frac{1}{\pi}P.V.\int_{-\infty}^{+\infty}\frac{\mu^{(2)}(s)}{\omega-s}ds;\\
\mu^{(2)}(\omega)=&-\frac{1}{\pi}P.V.\int_{-\infty}^{+\infty}\frac{\mu^{(1)}(s)}{\omega-s}ds.
\end{align}

The magnetic permeability of particle $\mu_c(\omega)$ can be
described by the Drude model \cite{SC1,AMRZ14,ADM6,FZS27}, i.e.
\begin{align}\label{drm}
\mu_c(\omega)=\mu_0\left(1-F\frac{\omega^2}{\omega^2-\omega_0^2+i\omega\tau^{-1}}\right),
\end{align}
where $\tau>0$ is the nanoparticle's bulk electron relaxation rate
($\tau^{-1}$ is the damping coefficient), $F$ is a filling factor, $\mu_0$ is permeability of free space,
and $\omega_0$ is a localized plasmon resonant frequency. In
particular, as $\omega\in\mathbb{R}$, if
\begin{align}\label{drm1}
(1-F)(\omega^2-\omega_0^2)^2-F\omega_0^2(\omega^2-\omega_0^2)+\tau^{-2}\omega^2<0,
\end{align}
we see the real part of $\mu_c(\omega)$ is negative, i.e., $\mu^{(1)}(\omega)<0$, which indicates
some singular property of particles.

\begin{theorem}\label{thm4.4}
Let $u$ be the solution of (\ref{2.3}). The magnetic permeability of particle $\mu_c(\omega)=\mu^{(1)}(\omega)+i\mu^{(2)}(\omega)$ satisfies the Kramer-Kronig relations (\ref{kk}) and
the Drude model (\ref{drm}). We find that, in the quasi-static regime, if $\lambda\left(\frac{|\mu_c(\omega)|^2}{\mu_m\mu^{(1)}(\omega)}\right)=\lambda_j$ for some $j\geq1$, and $a\cdot\nabla_z\mathcal {S}^{\alpha,0}[\varphi_j](z)\neq0$ for some $z\in Y\setminus\overline{D}$,
then, as the nanoparticle's bulk electron relaxation rate $\tau\rightarrow0^+$ or the filling factor $F\rightarrow0^+$, the near field energy
\begin{align}
\left\|\nabla u\right\|_{L^2(D)}\rightarrow\infty.
\end{align}
\end{theorem}
\begin{proof}
In Theorem \ref{thm4.3} (2), let $|\delta|=\frac{\mu^{(2)}(\omega)}{|\mu_c(\omega)|^2}=\frac{\tau F\mu_0\omega^3}{(\tau^2(\omega^2-\omega_0^2)^2+\omega^2)|\mu_c(\omega)|^2}$, we get the result.
\end{proof}

\begin{remark}
From Theorem \ref{thm4.4}, we see that, if the filling factor $F$ is fixed, then the rate of blows up of near field energy is $\mathcal{O}(\tau^{-1})$, (as $\tau\rightarrow0^+$). Likewise, fixing $\tau>0$, the rate of blows up of near field energy is $\mathcal{O}(F^{-1})$, (as $F\rightarrow0^+$).
\end{remark}

\begin{remark}
As application, we can solve a simple optical design problem: How to determine the magnetic permeability of nanoparticles
embedded into the optical device, such that the near field energy is maximized? In fact, according to Theorem \ref{thm4.4}, for fixed sufficiently small frequency $\omega$, giving the magnetic permeability of background medium $\mu_m$ and filling factor $F$, then we can obtain the nanoparticle's bulk electron relaxation rate $\tau$ by the following direct algorithm:

(i) Calculate the eigenvalue $\lambda_j$ ($j\geq1$) for $\alpha$-quasi-periodic Neumann-Poincar{\'e} operator $(\mathcal {K}^{-\alpha,0})^*$;

(ii) Using the Drude model (\ref{drm}), and find the bulk electron relaxation rate $\tau$ by solving the nonlinear equation
\begin{align*}
\lambda\left(\frac{|\mu_c(\tau)|^2}{\mu_m\mu^{(1)}(\tau)}\right)=\lambda_j,
\end{align*}
where the function $\lambda(t)$ is defined by (\ref{3.26}).

Meanwhile, by applying above algorithm, we can similarly compute the filling factor $F$. From the Drude model, it is easy to see that the magnetic permeability of nanoparticles are fully determined by the bulk electron relaxation rate and filling factor.

Moreover, the corresponding near field can be computed by (\ref{3.32}). Notice that the solution of (\ref{2.3}) and wave number $k_c$ are complex, by using Green's formula, the resonance near field energy can be written as
\begin{align*}
\left\|\nabla u\right\|_{L^2(D)}^2=\left((\Re{k_c})^2-(\Im{k_c})^2\right)&\left\|\mathcal {S}^{\alpha,k_c}[\phi]\right\|_{L^2(D)}^2\\&+\Re{\int_{\partial D}\mathcal {S}^{\alpha,k_c}[\phi]\overline{\left(-\frac{1}{2}\mathcal
{I}+(\mathcal {K}^{-\alpha,k_c})^*\right)[\phi]})d\sigma},
\end{align*}
which can be calculated by the Nystr{\"o}m method (see \cite{K31}).
\end{remark}

\section{Conclusion}
The mathematical analysis of plasmonic resonance of two dimensions photonic crystal in which embedded the double negative nanoparticles was given in this paper. Based on perturbation theorem, by applying the quasi-periodic layer potential method and the spectral theorem of quasi-periodic Neumann-Poincar{\'e} operator, we have deduced the quasi-static expansion of the near field and shown that when $\lambda\left(\frac{\mu_c}{\mu_m}\right)$ is close to some eigenvalue of quasi-periodic Neumann-Poincar{\'e} operator, the near field will enhance obviously.  Moreover, as the magnetic permeability of nanoparticles was described by the Drude model, the conditions under which the plasmonic resonance occurs, and the estimate for rate of blow up of near field energy with respect to nanoparticle's bulk electron relaxation rate and filling factor were also obtained.\\
\\
\textbf{Acknowledgments}\\
\\
We would like to thank Professor Habib Ammari and Dr. Hai Zhang for their useful discussions. The work
described in this paper was supported by the NSF of China (11301168) and the
Plan for the growth of young teachers of Hunan University (531107040658).

\section{Appendix}

The proof of Lemma \ref{lem1}:\\
\begin{proof}
It is sufficient to prove the last part of the Lemma \ref{lem1}, i.e. $\mathcal {S}^{\alpha,k}$ is an isomorphism from $H^{-\frac{1}{2}}(\partial D)$ to $H^{\frac{1}{2}}(\partial D)$. We first show the sesquilinear form $(\mathcal{S}^{\alpha,i}[\cdot],\cdot)_{\frac{1}{2},-\frac{1}{2}}$ on $H^{-\frac{1}{2}}(\partial D)$ is coercive, here, $\mathcal {S}^{\alpha,i}$ is the quasi-periodic single layer potential with the wave number $k=i$. In fact, let $v\in H^1(Y\setminus\partial D)$ be the quasi-periodic single layer potential given by
\begin{align*}
v(x)=\int_{\partial D}G^{\alpha,k}(x,y)\phi(y)d\sigma(y),\ \ \ \phi\in H^{-\frac{1}{2}}(\partial D),\ \ x\in Y\setminus\partial D.
\end{align*}
Clearly, $v$ satisfies Helmholtz equation in $D$ and $Y\setminus\overline{D}$ and the $\alpha-$ quasi-periodic boundary condition on $\partial Y$
\begin{align*}
v(1,x_2)=&e^{i\alpha_1}v(0,x_2),\ \ \ x_2\in(0,1),\\
v(x_1,1)=&e^{i\alpha_2}v(x_1,0),\ \ \ x_1\in(0,1),\\
\frac{\partial v}{\partial x_1}(1,x_2)=&e^{i\alpha_1}\frac{\partial v}{\partial x_1}(0,x_2),\ \ \ x_2\in(0,1),\\
\frac{\partial v}{\partial x_2}(x_1,1)=&e^{i\alpha_2}\frac{\partial v}{\partial x_2}(x_1,0),\ \ \ x_1\in(0,1).
\end{align*}
Setting $k=i$ in the definition of $v$, and applying Green's first identity, we find that
\begin{align*}
&\left(\mathcal{S}^{\alpha,i}[\phi],\phi\right)_{\frac{1}{2},-\frac{1}{2}}=\int_{\partial D}\left(\frac{\partial v}{\partial \nu}\bigg|_+-\frac{\partial v}{\partial \nu}\bigg|_-\right)\overline{v}d\sigma\\
=&\int_{Y\setminus\overline{D}}\left(|\nabla v|^2+|v|^2\right)\overline{v}dx-\int_{\partial Y}\overline{v}\frac{\partial v}{\partial \nu} d\sigma+\int_{D}\left(|\nabla v|^2+|v|^2\right)\overline{v}dx.
\end{align*}
Since $v$ is quasi-periodic function, we can see $\int_{\partial Y}\overline{v}\frac{\partial v}{\partial \nu} d\sigma=0$ (see \cite{AKL25} page 125), and then
\begin{align}\label{7.1}
\left(\mathcal{S}^{\alpha,i}[\phi],\phi\right)_{\frac{1}{2},-\frac{1}{2}}
=\int_{Y\setminus\overline{D}}\left(|\nabla v|^2+|v|^2\right)\overline{v}dx+\int_{D}\left(|\nabla v|^2+|v|^2\right)\overline{v}dx.
\end{align}
Furthermore, from the jump relations of $v$ and trace theorem, it follows that
\begin{align}\label{7.2}
\|\phi\|_{H^{-\frac{1}{2}}(\partial D)}=\left\|\frac{\partial v}{\partial \nu}\bigg|_+-\frac{\partial v}{\partial \nu}\bigg|_-\right\|_{H^{-\frac{1}{2}}(\partial D)}\lesssim\|v\|_{H^1(D)}+\|v\|_{H^1(Y\setminus\overline{D})}.
\end{align}
Thus, combining (\ref{7.1}) and (\ref{7.2}), it yields
\begin{align*}
\|\phi\|_{H^{-\frac{1}{2}}(\partial D)}\lesssim\left(\mathcal{S}^{\alpha,i}[\phi],\phi\right)_{\frac{1}{2},-\frac{1}{2}}.
\end{align*}

Notice that the integral kernel of $\mathcal{S}^{\alpha,k}-\mathcal{S}^{\alpha,i}$ is a $C^\infty$ function in a neighborhood of $\partial D\times \partial D$ and hence we conclude that $\mathcal{S}^{\alpha,k}-\mathcal{S}^{\alpha,i}$ is compact from $H^{-\frac{1}{2}}(\partial D)$ to $H^{\frac{1}{2}}(\partial D)$. Applying the Lax-Milgram lemma to the bounded and coercive sesquilinear form
\begin{align*}
a(\phi,\psi):=\left(\mathcal{S}^{\alpha,i}[\phi],\psi\right)_{\frac{1}{2},-\frac{1}{2}},\ \ \ \phi,\psi\in H^{-\frac{1}{2}}(\partial D).
\end{align*}
we imply that $(\mathcal{S}^{\alpha,i})^{-1}:H^{\frac{1}{2}}(\partial D)\rightarrow H^{-\frac{1}{2}}(\partial D)$ exists and is bounded.
Using the Theorem 5.14 in \cite{CC26} and compactness of $\mathcal{S}^{\alpha,k}-\mathcal{S}^{\alpha,i}$, it finds that $\mathcal{S}^{\alpha,k}$ is an isomorphism if and only if $\mathcal{S}^{\alpha,k}$ is injective. However, the injectivity of $\mathcal{S}^{\alpha,k}$ is proved in Lemma \ref{lem3.3}. Thus, the statements in
Lemma \ref{lem1} are proved.
\end{proof}


\end{document}